\documentclass[english, numberwithinsect, a4paper]{lipics-v2019}

\usepackage[utf8]{inputenc}

\usepackage{amsthm, amsmath, amssymb}

\usepackage{color}
\usepackage{tikz}
\usepackage{graphicx}

\newcommand{\N}{\mathbb{N}}

\DeclareMathOperator{\started}{started}

\tikzstyle{vertex}=[circle, draw, fill=black, inner sep=0pt, minimum width=4pt]
\tikzstyle{edge} = [line width = 1pt]

\authorrunning{H.~L. Bodlaender and M. van der Wegen}

\title{Parameterized Complexity of Scheduling Chains of Jobs with Delays}

\author{Hans L. Bodlaender}{Department of Information and Computing Sciences, Utrecht University,
the Netherlands}{H.L.Bodlaender@uu.nl}{https://orcid.org/0000-0002-9297-3330}{}
\author{Marieke van der Wegen}{Department of Information and Computing Sciences, Utrecht University, the Netherlands
\and 
Mathematical Institute, Utrecht University, the Netherlands}{M.vanderWegen@uu.nl}{https://orcid.org/0000-0003-0899-6925}{}
\date{}



\Copyright{Hans L. Bodlaender and Marieke van der Wegen}
\acknowledgements{We would like to thank Sukanya Pandey for helpful discussions. }

\hideLIPIcs
\nolinenumbers

\begin{document}
\maketitle

\begin{abstract}
In this paper, we consider the parameterized complexity of the following scheduling problem. We must schedule 
a number of jobs on $m$ machines, where each job has unit length, and the graph of precedence constraints
consists of a set of chains. Each precedence constraint is labelled with an integer that denotes the
exact (or minimum) delay between the jobs. We study different cases; delays can be given in unary and in binary, and
the case that we have a single machine is discussed separately. We consider the complexity of this problem
parameterized by the number of chains, and by the thickness of the instance, which is the maximum
number of chains whose intervals between release date and deadline overlap. 

We show that this scheduling problem with exact delays in unary is $W[t]$-hard for all $t$, when parameterized by the thickness, even when we have a single machine ($m=1$). When parameterized by
the number of chains, this problem is $W[1]$-complete when we have a single or a constant number of
machines, and $W[2]$-complete when the number of machines is a variable. The problem with minimum delays, given in unary, parameterized by the number of chains (and as a simple corollary, also when
parameterized by the thickness) is $W[1]$-hard for a single or a constant number of machines, and
$W[2]$-hard when the number of machines is variable. 

With a dynamic programming algorithm, one can show membership in XP for exact and minimum delays in unary, for any number of machines, when parameterized by thickness or number of chains. For a single machine, with exact delays in binary, parameterized by the number of chains, membership in XP can be shown with branching and solving a system
of difference constraints. For all other cases for delays in binary, membership in XP is open.
\end{abstract}

\newpage

\section{Introduction}
In this paper, we study a problem in the field of parameterized complexity of scheduling problems. Here,
we look at scheduling jobs with precedence constraints with exact or minimum delays, and assume that
jobs have unit length. We study one of the simplest types of precedence constraint graphs: we assume that the precedences form a collection of disjoint chains. Chains have a release date and deadline.
It is not hard to see (by a simple reduction from {\sc 3-Partition}) that this problem is NP-hard, even when all delays are 0. In this paper, we study the parameterized complexity of the problem, and look at two
different parameters: the number of chains, and the thickness of the instance --- that is, the maximum number of chains that have overlapping intervals from release time to deadline. We look at different variants: a constraint gives an exact bound or a lower bound on the delay between successive jobs; we can have one, a constant, or a variable number of machines, and the delays can be given in unary or binary notation.
If delays are given in unary, then each of the studied variants belongs to XP, and is hard for $W[1]$ (or 
classes higher in the $W$-hierarchy.) For one variation (see below), we also show membership in XP when delays are given in binary.
We call the studied problems {\sc Chain Scheduling with Exact Delays} and {\sc Chain Scheduling with Minimum Delays}, for details see Section~\ref{section:preliminaries}.

\subsection{Related literature.} Looking at variants of scheduling problems with special attention to
parameters (like the number of available machines) is a common approach in the rich field of study of
scheduling problems. Studying such parameterizations using techniques and terminology from the field
of parameterized algorithms and complexity was pioneered in 1995 \cite{BodlaenderF95}, but recently
receives growing attention, e.g. \cite{Bessy2019,Fellows2003,Mnich2015}.

The scheduling of chains of jobs (without delays) was studied by Woeginger~\cite{WOEGINGER} and
Bruckner~\cite{Brucker}, who gave respectively a 2-approximation algorithm, and a linear time algorithm
for two machines. General precedence graphs with delays between jobs was studied already in 1992 by Wikum et al.~\cite{wikum1994one}; this was followed by a large body of literature, studying different variations and approaches, including theoretical and experimental studies.

\subsection{A different interpretation}
The problem we studied can also be interpreted as another type of scheduling problem. Now, we let each chain represent {\em one} job: this job only occasionally needs to use the machine (or some specific resource) --- at the other time steps, the job is running but does not need a resource. E.g., if we have a chain with delays 2 and 3, then, we can interpret this as a job that needs a resource at its first, fourth, and eighth time step. The variant with minimum delays now has as interpretation that we allow pre-emption, i.e., jobs can be halted at some time steps and resumed later. 

\subsection{Our results}
\begin{table}[htb]
\begin{center}
\begin{tabular}{|l|l|l|l|}
\hline
     & parameter& exact delays & minimum delays  \\ \hline
    Single machine & thickness & $W[t]$-hard for all $t$     & $W[1]$-hard\\
     & chains & $W[1]$-complete & $W[1]$-hard \\ \hline
    Constant number & thickness & $W[t]$-hard for all $t$ & $W[1]$-hard\\
     of machines & chains & $W[1]$-complete& $W[1]$-hard \\ \hline
    Variable number & thickness & $W[t]$-hard for all $t$ & $W[2]$-hard\\
     of machines & chains & $W[2]$-complete& $W[2]$-hard \\ \hline
\end{tabular}
\end{center}
\caption{Hardness results for different variants of the problem. Exact and minimum delays are given in unary.}\label{table:hardnessresults}
\end{table}

In this paper, we give a number of hardness results, which are
summarized in Table~\ref{table:hardnessresults}. All variants are already hard when the (exact or minimum) delays are given in unary. We also give the following algorithmic results:

\begin{itemize}
    \item With a dynamic programming algorithm, one can show that the {\sc Chain Scheduling with Exact Delays} and {\sc Chain Scheduling with Minimum Delays} belong to XP, when delays are given in unary, and parameterized by either thickness or number of chains, for any number of machines.
    \item Combining branching with solving a set of difference constraints shows XP-membership of {\sc Chain Scheduling with Exact Delays} when parameterized by number of chains, for the case of one machines, when delays are given in binary.
\end{itemize}

For all other cases, the membership in XP when delays are given in binary is open.

\subsection{Organization of this paper}
In Section~\ref{section:preliminaries}, we give a number of preliminary definitions. Section~\ref{section:paramachines} gives hardness proofs for {\sc Chain Scheduling with Exact Delays}
when parameterized by the thickness. The complexity of {\sc Chain Scheduling with Exact Delays} parameterized by the number of chains is established in Section~\ref{section:chains}; a relatively simple modification then gives hardness for the corresponding problems with minimum delays. Section~\ref{sec:XP} 
gives our algorithmic results (membership in XP). Some conclusions are given in Section~\ref{section:conclusions}.

\section{Preliminaries}
\label{section:preliminaries}

We first describe the problems we study in more details. We have a number of identical machines $m$. In the paper, we study separately the cases that we have a single machine ($m=1$), the number of machines is
some fixed constant, or the number of machines is variable. 

On these machines, we must schedule $n$ jobs. Each job has unit length. On the set of jobs, we have
a collection of precedence constraints. Each precedence constraint is an ordered pair of jobs
$(i,i')$: it tells that job $i'$ cannot be started before job $i$ is completed. We say that $i$ is
a {\em direct predecessor} of $i'$, and $i$ is a {\em predecessor} of $i'$ if there is a directed
path from $i$ to $i'$ in the graph formed by the precedence constraints; $i'$ then is a 
{\em successor} of $i$.

The precedence constraints
have associated with them a {\em delay}, denoted $l_{i,i'}$: each delay is a non-negative integer. 
We study two variations of the problem: exact delays and minimum delays. If we consider exact (resp.\ minimum)
delays, then if constraint $(i,i')$ has delay $l_{i,i'}$ then job $i'$ must be started exactly (resp.\ at least) $l_{i,i'}$
time steps after job $i$ was finished. That is: when job $i$ starts at time $t$, then job $i'$ starts 
at time exactly (resp.\ at least) $t+ l_{i,i'} +1$. (Note that jobs run directly after each other, only if the delay is $0$.)
It is allowed to schedule a job on a different machine than its predecessor --- thus, we do not need to specify on what machine a job is running, but only ensure that at each time step, the number of
scheduled jobs is at most the number of available machines.

In this paper, we consider the case that the graph of the precedence constraints consists of a set of
chains. I.e., each job has at most one direct predecessor and at most one direct successor. Chains are the maximal sets of jobs that
are predecessors or successors of each other.

Each chain $C$ has a {\em release date} $r_C$ and a deadline $d_C$. We have that the first job in the chain cannot
start before time $r_C$ and the last job in the chain should be completed at or before time $d_C$.

In this paper, we consider the following two parameterizations of the problem. The first is the
number of chains, denoted by $c$. The second is the {\em thickness}, denoted by $\tau$, defined as follows.
We say that two chains {\em overlap}, when their intervals $[r_C,d_C)$ have a non-empty intersection.
We define the thickness $\tau$ to be the maximum size of a collection of chains that mutually overlap.
That is, for any time $t$, there are at most $\tau$ chains $C$ for which we have that $r_C \leq t$ and $d_C > t$. 

{\sc Chain Scheduling with Exact Delays} is the problem where we are given as input the set of jobs with
chains of precedence constraints, delays for each precedence constraint, release dates and deadlines of chains, and number of machines, and ask whether there exists a schedule that fulfills all the demands:
at each time step, the number of jobs scheduled is at most the number of machines; jobs in a chain
are not scheduled before the release date or after the deadline, and for each precedence constraint $(i,i')$
the delay between $i$ and $i'$ is exactly $l_{i,i'}$.

As said, we study several variants of this problem: delays can be given in unary or binary, the number
of machines can be 1, fixed or variable, and we can parameterize by the number of chains or by thickness.
If we require that the stated delays are lower bounds, we obtain the {\sc Chain Scheduling with Minimum Delays} problem: here, when we have a precedence contraint $(i,i')$ with delay $l_{i,i'}$, we must have that job $i'$ starts at least $l_{i,i'}$ time steps after job $i$ is finished.

\medskip

For the $W[t]$-hardness proofs, we use reductions from the following version of the
{\sc Satisfiability} problem. 
A Boolean formula is said to be {\em $t$-normalized}, if it is the conjunction of the disjunction of the conjunction of \ldots of literals, with $t$ alternations of AND's and OR's. 

The following parameterized problem was considered by Downey and Fellows \cite{DowneyF95}.

\begin{verse}
{\sc Weighted $t$-Normalized Satisfiability}\\
{\bf Given:} A $t$-normalized Boolean formula $F$ and a positive integer $k\in {\N}$. \\
{\bf Parameter:} $k$\\
{\bf Question:} Can $F$ be satisfied by setting exactly $k$ variables to true?
\end{verse}

\begin{theorem}[Downey and Fellows \cite{DowneyF95,DowneyF95II}]
For every $t\geq 2$,
{\sc Weighted $t$-Normalized Satisfiability} is $W[t]$-complete.
\end{theorem}

For the $W[1]$- and $W[2]$-completeness results, we use reductions from \textsc{Independent Set} and
\textsc{Dominating Set}. It is know that \textsc{Independent Set} is $W[1]$-complete \cite{DowneyF95II} and
\textsc{Dominating Set} is $W[2]$-complete \cite{DowneyF95}.

\section{Parameterization by thickness}
\label{section:paramachines}
In this section, we look at the \textsc{Chain Scheduling with Exact Delays} problem, when parameterized by the thickness $\tau$. We will show, for several variations, that the problem is hard for the class $W[t]$, for all integers $t$.

\subsection{Parallel machines}
We consider the version with $m$ parallel machines, where $m$ is part of the input.
The delays are assumed to be exact. 

We will give a reduction from \textsc{Weighted $t$-Normalized Satisfiablity}. Assume we have a $t$-normalized Boolean formula $F$ and integer $k$. Let $t'$ be the number of `levels' of disjunction.
We assume the variables of $F$
to be $x_0, \ldots, x_{n-1}$. 

We make an instance of the \textsc{Chain Scheduling with Exact Delays} problem, with $m = k+t'$ machines and thickness $\tau = 2k+t'$. 

An {\em element} of the formula is either a literal, or a disjunction or conjunction of smaller elements. We will first assign to each element a size $s$, and then to each element an integer interval. To each element, we 
associate an {\em interval size}; the interval size of formula $F'$ is denoted by $s(F)$.

The {\em interval size} of a literal (i.e., a formula of the form $x_i$ or $\neg x_i$) is 
$2n$. The interval size of a conjunction is the sum of the size of the terms,
i.e., $s(F_1 \wedge F_2 \wedge \cdots \wedge F_q) = \sum_{i=1}^q s(F_i)$.
For each disjunction $F'$ of $q$ terms, its size is $2q+1$ times the maximum size of its terms:
define $s_{\max}(F')= \max_{1\leq i\leq q} s(F_i)$, and then 
$s(F_1 \vee F_2 \vee \cdots \vee F_q) = (2q+1) \cdot s_{\max}(F') $.

To each element $F'$ of $F$ we assign an integer interval $[\ell(F'),r(F')]$ with $s(F') = r(F') - \ell(F')$. We will do this top-down: first we assign an interval to $F$, then we define a subinterval for every term of $F$, etc. 

To $F$, we assign the interval $[n,n+s(F)]$. To elements of a conjunction and disjunction, we assign subintervals of the intervals assigned to the conjunction of disjunction, in such a way that these intervals have the same nesting as the elements in the formula.

Consider an element $F'$ that is the conjunction $F_1 \wedge F_2 \wedge \cdots \wedge F_s$.
Then assign $F_1$ the interval $[\ell(F'),\ell(F')+s(F_1)]$; $F_2$ the interval $[\ell(F')+s(F_1),\ell(F')+s(F_1)+s(F_2)]$,
etc. I.e., $F_i$ is assigned the interval $[\ell(F')+ \sum_{j=1}^{i-1} s(F_j), \ell(F')+ \sum_{j=1}^{i} s(F_j)]$.

Suppose element $F'$ is the disjunction $F_1 \vee F_2 \vee \cdots \vee F_s$.
The construction is similar to that of conjunctions, but now we assign each term the same length interval and keep unused intervals between
the terms. Recall that $s_{\max}(F') = \max_{1\leq i\leq s} s(F_i)$. 
Assign to $F_i$ the interval $[\ell(F') + (2i-1)\cdot s_{\max}(F'), \ell(F')+ 2i\cdot s_{max}(F')]$. 

Note the nesting of intervals, and that we assigned to each element an interval equal to its size.
Also note that we can compute all intervals and sizes in polynomial time.

We now can describe the jobs, precedence constraints and release dates and deadlines.

For each $i$, $1\leq i\leq k$, we start a chain $c_i$. Each of those chains start with a job and then a delay of $n-1$. The first job of the chain is released at time 0.  We will add jobs and specify delays between jobs in the chain such that the total processing time including the delay times is $n+s(F) + 1$. Set the deadline of those chains to $2n +s(F)$, so that the first job can start at times $0,1, \ldots, n-1$.

These chains reflect the variables that are set to true; more precisely, when the first job of one of the chains starts at a time $i$, then this corresponds to setting $x_i$ to true. We call these the {\em true variable chains}. 

To prevent two chains selecting the same variable, we 
add $m-1$ chains, each with $n$ jobs with delay $0$, release date $0$ and deadline $n$. We call those chains \emph{fill chains}. Those chains have to be scheduled from time $0$ until time $n$. This implies that at each time $0, 1, \ldots, n-1$ at most one other job can be schedules, thus at most one true variable chain starts. Hence, the true variable chains select exactly $k$ variable to be set to true. 

We will now extend the true variable chains. Consider the interval $[n,n+s(F)]$ from left to right. 
\begin{itemize}
    \item For each timestep that we encounter that is not part of an interval that corresponds to a literal, we add a delay of $1$ at the end of the chain. 
    \item For each interval $[\ell(F'), r(F')]$ that corresponds to a positive literal $x_i$, we add the following gadget to the chain: $n-1-i$ jobs with delay $0$, then a delay of $1$, then $i$ jobs with delay $0$ and then a delay of $n$.  (See Figure \ref{fig:variable-gadgets}.)
    Notice that no job is scheduled from $\ell(F') +n-1$ until $\ell(F') +n$ if the chain starts at time $i$, and there is a job scheduled at this time otherwise. 
    \item For each element $F'$ of $F$ that is a negative literal $\neg x_i$, we make the following gadget:  a delay of $n-1-i$, a job, then a delay of $i$, and then a delay of $n$.  (See Figure \ref{fig:variable-gadgets}.)
    Notice that a job is scheduled from $\ell(F') +n-1$ until $\ell(F') +n$ if the chain starts at time $i$, and there is no job scheduled at this time otherwise. 
\end{itemize}
Add one job at the end of the chain. 
Notice that the total processing time of those chains is indeed $n + s(F) + 1$. 
\begin{figure}
\includegraphics{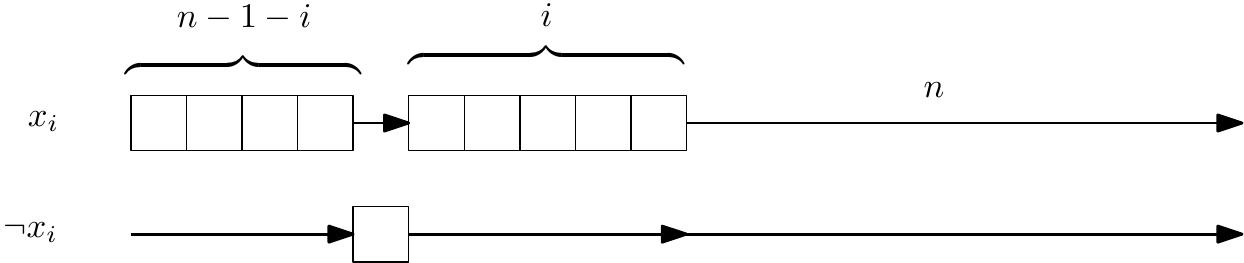}
\caption{The variable gadgets.} \label{fig:variable-gadgets}
\end{figure}

To check whether variables are true, we 
add some chains that consist of a single job. We call those chains \emph{variable check chains}. 
\begin{itemize}
    \item For each element $F'$ of $F$ that consists of a single positive literal (i.e., is of the form $x_i$),
    we make a chain with one job, that is released at time $\ell(F')+n - 1$ and has deadline $\ell(F') + n$. 
    \item For element $F'$ of $F$ that consists of a single negative literal (i.e., is of the form $\neg x_i$), we make a $k$ chains with one job, release date $\ell(F')+n - 1$ and deadline $\ell(F') + n$.
\end{itemize}
The intuition behind this construction is as follows: suppose that we have $k$ machines. For each element $F'$ of $F$ that is of the form $x_i$, there is one job scheduled from $\ell(F')+n - 1$ until $\ell(F') + n$. So for at least one of the true variable chains, we need that no job of this chain to be scheduled from $\ell(F')+n - 1$ until $\ell(F') + n$. This means that one of the true variable chains starts at time $i$. 
For each element $F'$ of $F$ that is of the form $\neg x_i$, there are $k$ job scheduled from $\ell(F')+n - 1$ until $\ell(F') + n$. So for none of the true variable chains we can schedule a job of this chain from $\ell(F')+n - 1$ until $\ell(F') + n$. This means that none of the true variable chains starts at time $i$. 
The other $t'$ machines take care of the disjunctions. 

For each element $F' = F_1 \vee F_2 \vee \cdots \vee F_q$ of $F$ that is a disjunction, we make one chain. This chain has $3\cdot s_{\max}(F')$ jobs with delay $0$. 
The chain will be released at time $\ell(F')$ and has deadline $r(F')$. We call those chains \emph{disjunction chains}. Notice that for every element $F'$ of $F$ that is a literal, there are exactly $t'$ disjunction chains that overlap the interval $[l(F'), r(F')]$, that is, there are exactly $t'$ disjunction chains $C$ with release time at most $l(F')$ and deadline at least $r(F')$.

We now have specified all jobs and the machines they run on. Note that the thickness of this construction is $2k+t'$.

\begin{claim}
If $F$ is satisfiable by setting exactly $k$ variables to true, then the \textsc{Chain Scheduling with Exact Delays} scheduling problem has a solution.
\end{claim}

\begin{proof}
Suppose $F$ is satisfiable by making variables $x_{i_1}, \ldots, x_{i_k}$ true. For each $j$, with
$1\leq j\leq k$, we let one true variable chain start at time $i_j$.

First we introduce a notion \emph{satisfying}, intuitively, this will be the elements that make $F$ true. We define this top-down. First we call $F$ satisfying. For each element $F'$ of $F$:
\begin{itemize}
    \item If $F'$ is satisfying and $F'$ is a conjunction $F' = F_1 \wedge \cdots \wedge F_q$, then all terms $F_i$ are satisfying. 
    \item If $F'$ is satisfying and $F'$ is a disjunction $F' = F_1 \vee \cdots \vee F_q$, then at least one $F_i$ is satisfied, say $F_{j}$. We say that $F_{j}$ is satisfying, and the other term $F_{i}$ with $i\neq j$ are not satisfying. 
    \item If $F'$ is not satisfying, all its terms are not satisfying. 
\end{itemize}

For each element $F'$ of $F$ that is a disjunction $F' = F_1 \vee \cdots \vee F_q$, consider the disjunction chain $C$ associated with this element. If $F'$ is not satisfying, then start this chain $C$ arbitrarily, say at its release time. If $F'$ is satisfying, let $F_j$ be its term that is satisfying. 
Now, start this chain $C$ at time $\ell(F')+ (2j - 2) s_{\max}(F')$, where $s_{\max}(F')$ is again the maximum interval size of the terms $F_i$.

Now we will verify that is a feasible schedule, that is, that we never use more than $m$ machines. 

First consider a time step $\alpha$ from $i$ to $i+1$, with $0\leq i\leq n-1$. At this time step there are $m-1$ jobs of fill chains scheduled. Since the variables $x_{i_1}, \ldots, x_{i_k}$ are different, the true variable chains start at different times. Hence, there is at most one job of a true variable chain scheduled at $\alpha$. Thus, there are at most $m$ jobs scheduled at time $\alpha$. 

Consider a time step $\alpha$ from $i$ to $i+1$ with $i \geq n$, such that $i \neq \ell(F')+n - 1$ for all elements $F'$ of $F$ that are literals. Notice that at those times no jobs of fill chains and variable check chains are scheduled. There are at most $k$ jobs of true variable chains scheduled at time step $\alpha$, since there are only $k$ true variable chains. And there are at most $t'$ jobs of disjunction chains scheduled, since there are $t'$ levels of disjunction. It follows that there are at most $k + t'= m$ jobs scheduled at time step $\alpha$.  

Consider a time step $\alpha$ from $i$ to $i+1$ with $i \geq n$, such that $i = \ell(F')+n - 1$ for an element $F'$ of $F$ that is a literal. We distinguish three cases. 

Suppose that $F'$ is satisfying, and $F'$ is a positive literal $F' = x_j$. Then we know that there are $t'$ machines used for the disjunction chains. Besides, there is one machine used for the variable check chains. Since we know that $F'$ is satisfying, $x_j$ is true. Thus one of the true variable chains starts at time $j$. So this chain has no job scheduled from $i$ until $i+1$. So at most $k-1$ true variable chains have a job scheduled from $i$ until $i+1$. In total there are at most $t' + 1 + k-1 = m$ machines used. 

Suppose that $F'$ is satisfying, and $F'$ is a negative literal, say $F'= \neg x_j$. Then we know that there are $t'$ machines used for the disjunction chains. There are $k$ machines used for the variable check chains. Since we know that $F'$ is satisfying, $x_j$ is false. Thus none of the true variable chains starts at time $j$. Hence, no true variable chain has a job scheduled from $i$ until $i+1$. In total there are at most $m$ machines used from $i$ until $i+1$. 

Suppose that $F'$ is not satisfying. Let $F''$ be the last satisfying element on the path from $F$ to $F'$. The disjunction chain that corresponds to $F''$ has no job that is scheduled form $i$ until $i+1$. Hence, there are at most $t'- 1$ machines used for the disjunction chains from time $i$ until $i+1$. 
It follows that there are $k+1$ machines for the true variable chains and the variable check chains. Since all true variable chains start at different times, those chains together have $k+1$ jobs that can be scheduled from $i$ until $i+1$. We conclude that we used at most $m$ machines from time $i$ until $i+1$. 
\end{proof}

\begin{claim}
Suppose the \textsc{Chain Scheduling with Exact Delays} scheduling problem has a solution, then $F$ can be satisfied by setting exactly $k$ variables to true.
\end{claim}

\begin{proof}
For each true variable chain, if it starts at time $i$, then set $x_i$ to true. All other variables are set
to false, i.e, $x_i$ is true, if and only if there is a true variable chain that starts at time $i$.

Note that all variable chains must start at different times, and they must start at times $0, 1, \ldots n-1$.
If two of these start at the same time $i$, then both have a job scheduled from $i$ until $i+1$, but there are $m-1$ fill chains that have a job scheduled from $i$ until $i+1$ as well, which is a contradiction with the total number of machines. Thus, we have set exactly $k$ variables to true.

We claim that this setting makes $F$ true. We define elements of $F$ to be \emph{demonstrated} in the following way, recursively. We say that $F$ is demonstrated.
\begin{itemize}
    \item If a conjunction is demonstrated, then all its terms are demonstrated. 
    \item If a disjunction $F' = F_1 \vee \cdots \vee F_q$ is demonstrated, then we consider the corresponding disjunction chain $C$. If for every time step in $[l(F_j), r(F_j)]$ a job of $C$ is scheduled, we say that $F_j$ is demonstrated.
\end{itemize}
After this top-down definition of demonstrated elements, we will now inductively show bottom-up that
demonstrated elements are satisfied by the setting of variables described above.

Consider a demonstrated element $F'$ that is a literal. By the definition of demonstrated, $t'$ of the disjunction chains have a job scheduled at each time step of $[l(F'), r(F')]$. Consider the time step from $l(F') + n -1$ until $l(F') + n$. We know that $t'$ machines process a job of a disjunction chain at this time step, so at most $k$ machines process jobs of true variable chains and variable check chains.  

Suppose that $F'$ is a positive literal, say $x_i$. Notice that one machine processes a check variable chain from $l(F') + n -1$ until $l(F') + n$. So at most $k-1$ machines process a job of a true variable chain. Hence there must be at least one true variable chain that does not have a job scheduled at this time. By construction of the true variable chains, such a chain must start at time $i$; and thus we set $x_i$ to true, i.e., our setting satisfies
the formula consisting of the single positive literal $x_i$.

If $F'$ is a negative literal $\neg x_i$, $k$ machines are processing check variable chains from $l(F') + n -1$ until $l(F') + n$; thus no true variable chain can have a job scheduled at this step. By construction
of the true variable chains, this implies that no true variable chain starts at time $i$, and thus $x_i$ is set to false. Hence, $F'$ is satisfied. 
 
Now consider a demonstrated element $F'$ and assume, by induction, that for all its terms $F_i$ holds: if $F_i$ is demonstrated, then $F_i$ is true. 

Suppose that $F'$ that is a conjunction. All its terms are demonstrated, and by
the induction hypothesis, all its terms are true. Thus the conjunction is also true.

Suppose that $F'$ is a disjunction. Since the corresponding disjunction chain has length $3s_{\max}(F')$, there is at least one $F_j$ such that for every time step in $[l(F_j), r(F_j)]$ a job of $C$ is scheduled. By definition, this $F_j$ is demonstrated. By induction, $F_j$ is true, and thus $F'$ is true.

We conclude that $F$ is satisfied. 
\end{proof}

We now have shown:

\begin{theorem}
The \textsc{Chain Scheduling with Exact Delays} problem, parameterized by the thickness $\tau$, is $W[t]$-complete  for all $t\in \N$. 
\label{thm:thickness-m-machines-Wt}
\end{theorem}

\subsection{Single machine}
Again, assume the delays are exact. 
We now show that the problem stays hard when there is only one machine.  

\begin{theorem}
The \textsc{Chain Scheduling with Exact Delays} problem, parameterized by the thickness $\tau$, is $W[t]$-complete  for all $t\in {\N}$, when only $1$ machine is available.
\label{theorem:thickness-1-machine-Wt}
\end{theorem}

Let $\tau$ be the thickness of the original instance. 
The main idea of the transformation is to replace each time step on $m$ machines by $\tau$ time steps on a single machine. Every chain will be assigned a number $i \in \{0, 1, \ldots, \tau - 1\}$ and its jobs will be scheduled at times $i \pmod \tau$, this will make sure that at every timestep only one job is scheduled. For every interval $[i\tau, (i+1)\tau]$, we will have $\tau-m$ chains that have one job; this ensures that at at most $m$ time steps of the interval a job of a regular chain is scheduled. We now proceed with the formal description. 

We transform from the case with $m$ machines (Theorem~\ref{thm:thickness-m-machines-Wt}). 
Suppose we have an instance with $m$ machines. For every interval $[i\tau, (i+1)\tau]$, we add $\tau-m$ additional chains, with a single job, release date $i\tau$ and deadline $(i+1)\tau$. We call those chains {\em extra}.

We copy the chains from the given instance, except that: 
\begin{itemize}
    \item If a chain has release date $\alpha$, then it now has release date $\alpha\cdot \tau$.
    \item If a chain has deadline $\beta$, then it now has deadline $\beta\cdot \tau$. 
    \item Every delay $d$ is replaced by a delay $\tau d + \tau - 1$. 
\end{itemize}
We call these chains {\em regular}.

\begin{claim}
Suppose we have a solution for the transformed instance with one machine. Then we have a solution for the original instance with $m$ machines.
\end{claim}

\begin{proof}
For each regular chain, let it start at time $\lfloor t/\tau \rfloor$, when its transformed instance starts at time $t$. This implies that for every job in the chain it will start at time $\lfloor t/\tau \rfloor$, when its transformed instance starts at time $t$. We know that for each time interval $[t\tau, (t+1)\tau]$, $\tau-m$ steps are used for extra jobs, so $m$ time steps are available for jobs of regular chains. 
Thus, at every time step $t$ in the original instance, at most $m$ jobs are scheduled.
\end{proof}

\begin{claim}
Suppose we have a solution for the original instance with $m$ machines. Then we have
a solution for the transformed instance with $1$ machine.
\end{claim}
\begin{proof}
We start with assigning a number $0, 1, \ldots, \tau-1$ to every chain such that for every time step all the chains that overlap this time step have different number. We denote this assignment by $c$. We can do this as follows: go through time $0, 1, 2, \ldots$, and  every time a chain is released, assign a number that is currently unused, if a deadline passes, the number of the corresponding chain becomes available again. We can do this with $\tau$ numbers, since by definition for every timestep there are at most $\tau$ chains that overlap this timestep.  

For every regular chain $C$, let $C$ start at time $t\tau + c(C)$, where $t$ is the starting time of the corresponding original chain. Then chains of thickness number $i$ only have jobs starting at times $t$ with $t \equiv i \pmod \tau$.
For every time $t$, all jobs that are scheduled to start at time $t$ in the original instance, will now be scheduled in the interval $[t\tau, (t+1)\tau]$ in the transformed instance. Those jobs are in different chains and those chains are assigned different numbers. Thus those jobs are scheduled at different times in the tranformed instance.  

In the original instance, at each time $t$ at most $m$ chains have a job scheduled to start at $t$, so, in the transformed instance, for each interval $[t\tau, (t+1)\tau]$, at most $m$ regular chains have a job scheduled in this interval. Thus, we can schedule the
$\tau-m$ extra chains in this time interval at the time steps where no job of a regular chain is scheduled.
\end{proof}

As the transformation can be carried out in polynomial time, Theorem~\ref{theorem:thickness-1-machine-Wt} follows
from the transformation and Theorem~\ref{thm:thickness-m-machines-Wt}.

\subsection{Constant number of parallel machines}

We can easily transform the single machine instance to an instance with a constant number $m$ of parallel machines. Let $T$ be the maximum deadline of all chains. Introduce $m-1$ new chains with $T$ jobs each and $0$ delays. The $m-1$ new machines will be processing those $m-1$ new chains, while the original machine processes the original chains. We conclude the following result. 
\begin{theorem}
The \textsc{Chain Scheduling with Exact Delays} problem with a fixed number of machines, parameterized by the thickness $\tau$, is $W[t]$-complete  for all $t\in {\N}$.
\end{theorem}

\subsection{Minimum delays}
The proof above seems not to modifiable to the \textsc{Chain Scheduling with Minimum Delays} problem.
In Section~\ref{subsection:minimumtimeschains}, we show that \textsc{Chain Scheduling with Minimum Delays}, parameterized
by the number of chains is $W[1]$-hard when $m=1$, and $W[2]$-hard when the number of
machines $m$ is variable. As the thickness is at most the number of chains, it follows that
\textsc{Chain Scheduling with Minimum Delays}, parameterized by the thickness is $W[1]$-hard for one machine,
and $W[2]$-hard for a variable number of machines. Membership in $W[1]$ or $W[2]$ is open, however.

\section{Parameterization by the number of chains}
\label{section:chains}
We now give the complexity results when we use the number of chains as parameter. 
In Sections~\ref{subsection:singlemachine}, \ref{subsection:constantmachines}, and 
\ref{subsection:variablemachines}, we consider \textsc{Chain Scheduling with Exact Delays},
with the number of machines respectively 1, a constant, or variable. In Section~\ref{subsection:minimumtimeschains}, we consider \textsc{Chain Scheduling with Exact Delays}.

\subsection{Single machine}
\label{subsection:singlemachine}
In this section, we consider the variant where the number of chains $c$ is a parameter, and at each step in time, there is one machine available. 
We assume that the delays are exact. 

\begin{theorem}
The \textsc{Chain Scheduling with Exact Delays} problem, parameterized by the number of chains $c$ is $W[1]$-complete, when there is one machine. 
\label{thm:chains-W1}
\end{theorem}

Theorem~\ref{thm:chains-W1} is proven by two transformations: from and to \textsc{Independent Set} with standard parameterization.

\begin{lemma}
The \textsc{Chain Scheduling with Exact Delays} problem with one machine, parameterized by the number of chains $c$, is $W[1]$-hard. 
\end{lemma}
\begin{proof}
A set of integers $S$ is said to be a \emph{Golomb ruler} if all differences $a-b$ of two elements $a, b \in S$ are unique, that is, $s_1 - s_2 \neq s_3 - s_4$ for $s_1, s_2, s_3, s_4 \in S$ with $s_1 \neq s_2$ and $s_3 \neq s_4$.

Erd\"{o}s and Tur\'{a}n \cite{ErdosT41} gave the following explicit construction of a Golumb ruler. Let $p>2$ be
a prime number. Then the set $\{2pk + (k^2 \bmod p) ~|~ k\in \{0, 1, \ldots, p-1 \}\}$ is a Golomb ruler with $p$ elements.
We can build a Golumb ruler of size $n$ in $O(n \sqrt{n})$ time: with help of the Sieve of Eratosthenes, we find a 
prime number $p$ between $n$ and $2n$ (such a number always exist, by the classic postulate of Bertrand (see \cite{Aigner2001}), and then follow the Erd\"{o}s-Tur\'{a}n-construction with this value of $p$, and take the first $n$ elements of this set.
Notice that the elements in this set are smaller than $4n^2$.

Suppose we have an input of \textsc{independent set} $G=(V,E)$ and $k$. Assume $V = \{v_1, \ldots, v_n\}$ and let $m$ be the number of edges. 
First, build a Golumb ruler $S^n$ of size $n$. Denote the elements by $s_0 \leq s_1 \leq \ldots \leq s_{n-1}$. Notice that $s_0 = 0$. Write $c_0 = s_{n-1} + 1$. 

We will construct an instance of \textsc{Chain Scheduling with Exact Delays}. 
We will make $k+1$ chains. We call one chain the \emph{start time forcing chain}, the other $k$ chains are the \emph{vertex selection chains}.

The start time forcing chain has release date $1$, deadline $c_0$ and total execution time (including delays) $c_0-1$. I.e., it must start at time 1. The chain will have a job starting at every time in $[1, c_0-1]$ except the times $s_1, s_2, \ldots, s_{n-1}$. 

The vertex selection chains have release date $0$. They have deadline $c_0+T -1$ and total execution time $T$, where $T = (m\cdot k(k-1))(2c_0 + 1) + 1$. So, they can start at times $0, 1, \ldots, c_0-1$.
The vertex selection chains start with a job and then a delay of $c_0 - 1$. Note that as a result of this, in order not
to conflict with the start time forcing chain, they have to start at an element of $S^n$.

Now, for each edge $\{v_i,v_j\}\in E$, and each ordered pair of vertex selection chains $C_a, C_b$ with $a,b \in \{1,2,\ldots, k\}$ we dedicate an interval $I_{v_i, v_j, C_a, C_b}$ of $2c_0 +1$ time steps.
More precisely, we have $m\cdot k(k-1)$ intervals $[c_0 + i(2c_0+1), c_0 + (i+1)(2c_0 + 1)]$ for $i= 0, 1, \ldots, m\cdot k(k-1) - 1$. And to each interval we assign a unique label $I_{v_i, v_j, C_a, C_b}$ where $v_iv_j \in E$ and $a, b \in \{1, 2, \ldots, k\}$. In the interval $I_{v_i, v_j, C_a, C_b}$ we will check whether the chains $C_a$ and $C_b$ did not select the edge $v_iv_j$, that is, whether $C_a$ does not start at $s_i$ or $C_b$ does not start at $s_j$. 

We will now extend the vertex selection chains. Consider the interval $[c_0,c_0 + c_0 + (m\cdot k(k-1))(2c_0 + 1)]$ from left to right. For each interval $I_{v_i, v_j, C_a, C_b}$ that we encounter, we extend the vertex selection chains as follows. 
\begin{itemize}
    \item For each chain $C$, with $C \neq C_a$, $C\neq C_b$, add a delay of $2c_0 + 1$. 
    \item Add the following gadget to the chain $C_a$: a delay of $c_0 - s_i$, a job, and then a delay if $c_0 + s_i$. 
    \item Add the following gadget to the chain $C_b$: a delay of $c_0 - s_j$, a job, and then a delay if $c_0 + s_j$. 
\end{itemize}
Add one job at the end of all chains. See Figures \ref{fig:Golomb-ruler}, \ref{fig:interval-chains} and \ref{fig:example} for an example of the construction and a feasible schedule. 

\begin{figure}
    \centering
    \begin{tikzpicture}
    \node[vertex, label={$v_2$}] (a) at (0,0) {};
    \node[vertex, label={$v_1$}] (b) at (.86,.5) {};
    \node[vertex, label={$v_3$}] (c) at (.86,-.5) {};
    \draw (a) --(b);
    \draw (a) -- (c);
    \end{tikzpicture}
    \caption{An example graph $G$. If we pick $p = 3$, the corresponding Golomb ruler is $\{0, 7, 13\}$. }
    \label{fig:Golomb-ruler}
\end{figure}
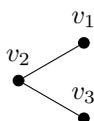
\begin{figure}
\includegraphics[width=\textwidth]{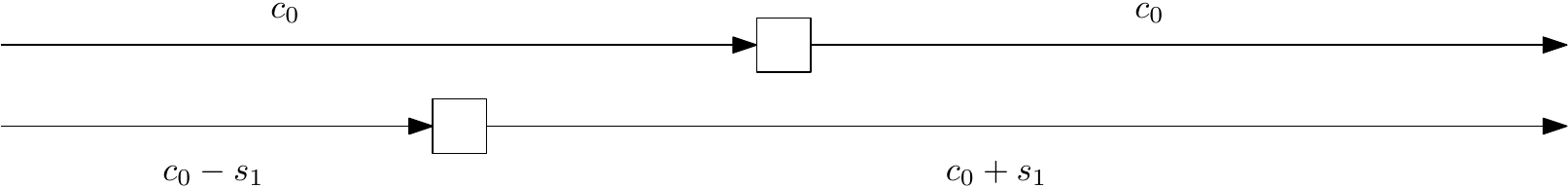}
\caption{The part of the chains $C_a$ and $C_b$ for the interval $I_{v_1, v_2, C_a, C_b}$ for the graph in Figure \ref{fig:Golomb-ruler}. } \label{fig:interval-chains}
\end{figure}
\begin{figure}
\includegraphics[width=\textwidth]{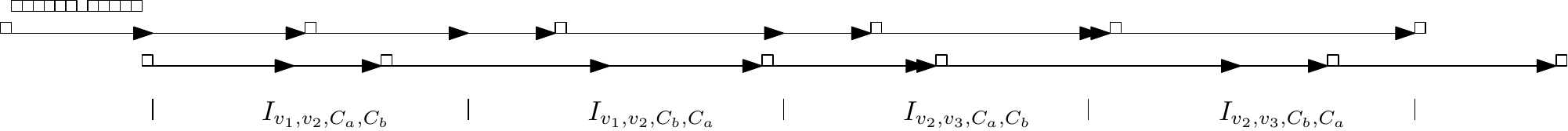}
\caption{The instance of the scheduling problem constucted from the graph in Figure \ref{fig:Golomb-ruler} and $k = 2$.} \label{fig:example}
\end{figure}

We claim that there is a feasible schedule, if and only if $G$ has an independent set of size at least $k$.

\begin{claim}
If $G$ has an independent set of size at least $k$, then there is a feasible schedule. 
\end{claim}
\begin{claimproof}
Suppose $v_{i_1}, \ldots, v_{i_k}$ form an independent set in $G$. Let the vertex selection chain $C_a$ start at times $s_{i_a}$ for $a = 1, 2, \ldots, k$. 

Notice that for the first $c_0$ time steps, at most one resource is used. The same holds for the last $c_0$ time steps. We show that this is a feasible scheduling by contradiction. Suppose that there are two resources needed at some time step $\beta$, and that $\beta$ is in the interval $I = I_{v_i, v_j, C_a, C_b}$ for checking whether the chains $C_a$ and $C_b$ do not select the edge $v_iv_j$. Write $I = [\ell(I), r(I)]$. 
Notice that job of $C_a$  in the interval $I$ starts at time $\ell(I) + c_0 - s_i + s_{i_a}$. And the job of $C_b$ in the interval $I$ starts at time $\ell(I) + c_0 - s_j + s_{i_b}$. Thus  $\ell(I) + c_0 - s_i + s_{i_a} = \beta = \ell(I) + c_0 - s_j + s_{i_b}$. Equivalently, $s_{i_a} - s_i = s_{i_b} - s_j$. Since $S^n$ is a Golomb ruler, it follows that either $s_{i_a} = s_i$ and $s_{i_b} = s_j$ or $s_{i_a} = s_{i_b}$ and $s_i = s_j$. 

In the first case, $s_{i_a} = s_i$ and $s_{i_b} = s_j$, we see that $v_i = v_{i_a}$ and $v_j = v_{i_b}$. But there is no edge $v_{i_a}v_{i_b}$, since $v_{i_a}$ and $v_{i_b}$ are in an independent set. This yields a contradiction. 

In the second case, $s_{i_a} = s_{i_b}$ and $s_i = s_j$, we see that $v_{i_a} = v_{i_b}$, but this yields a contradiction with the fact that the vertices of the independent set are distinct. 

We conclude that this schedule always uses at most one machine. 
\end{claimproof}

\begin{claim}
If there is a feasible schedule, then $G$ has an independent set of size at least $k$. 
\end{claim}
\begin{claimproof}
Suppose that there is a feasible schedule. Notice that the time of the first step of a vertex selection chains must be an element of $S^n$, 
otherwise the job conflicts with the start time forcing chain. 
Now, set $W = \{v_i \mid \textrm{there is a vertex selection chain that starts at time } s_i\}$. Notice that $|W|=k$, as otherwise
two vertex selection chains start at the same time, and conflict with each other for their first job. 

We prove that $W$ is an independent set by contradiction. Suppose that there exists an edge $v_iv_j$, with $v_i, v_j \in W$. Let $C_a$ be the chain that starts at $s_i$ and $C_b$ the chain that starts at $s_j$. Now consider the interval $I = I_{v_i, v_j, j_a, j_b}$, and write $I = [\ell(I), r(I)]$. By the construction of the chains, it follows that $C_a$ starts its job of the interval $I$ at time $\ell(I) + c_0 - s_i + s_i = \ell(I) + c_0$. The job of $C_b$ in the interval starts at time $\ell(I) + c_0$ as well. This yields a contradiction. We conclude that $W$ is an independent set.
\end{claimproof}
This shows that the \textsc{Chain Scheduling with Exact Delays} problem with a single machine, parametrized by the number of chains, is $W[1]$-hard. 
\end{proof}

\begin{lemma}
The \textsc{Chain Scheduling with Exact Delays} problem with one machine, parameterized by the number of chains $c$ is in $W[1]$. 
\end{lemma}
\begin{proof}
We show that the problem belongs to $W[1]$ by a transformation to
\textsc{Independent Set}. Suppose we are given an instance of the \textsc{Chain Scheduling with Exact Delays} problem with a single machine. 
We now build a graph as follows: for each chain $C$ and each possible starting time $t$ of this chain, we take one vertex $v_{C,t}$. We now add edges as follows: each pair of vertices that represent the same chain with different starting times is adjacent (i.e., for each chain, its vertices form a clique). For each
pair of vertices $v_{C,t}$, $v_{C',t'}$, $C\neq C'$, we add an edge if and only if starting chain $C$ on time $t$ and starting chain $C'$ on time $t'$ will cause that there is a time step where a job of both chains is scheduled.

It is not hard to see that the given instance of the \textsc{Chain Scheduling with Exact Delays} problem has a solution, if and only if the constructed graph has an independent set of size at least $k$ (with $k$ the number of chains). Indeed, if there is an independent set of size at least $k$, then for every chain $C$ there is a vertex $v_{C,t}$ in the independent set, since for every chain $C$ we can have at most one vertex $v_{C,t}$ in an independent set. Scheduling chain $C$ at time $t$ gives a feasible schedule. On the other hand, if we have a feasible schedule, then the set $\{v_{C,t} \mid C \textrm{ a chain }, t \textrm{ its starting time in the schedule}\}$ is an independent set of size $k$.
\end{proof}

\subsection{Constant number of parallel machines}
\label{subsection:constantmachines}
If we assume that there are $m$ machines, but $m$ is considered to be a
constant (i.e., not a fixed parameter; $m$ is part of the problem description), then the problem is 
$W[1]$-complete. 

\begin{theorem}
The \textsc{Chain Scheduling with Exact Delays} problem with $m$ machines, parameterized by the number of chains, is $W[1]$-complete. 
\label{thm:chains-constant-machines-W1}
\end{theorem}

Hardness follows easily from the case that $m=1$: add $m-1$ chains with maximum length that consist of jobs with 0 delay. 

Membership follows by formulating the problem as a \textsc{Weighted CNF-SAT} problem, i.e., giving a CNF-SAT formula 
that has a solution with at most $k$ (parameter) variables true, if and only if the scheduling problem has a solution.

We have a variable for each chain $C$ and each time $t$ that it can start, say $x_{C,t}$.

We have a number of clauses that force that each chains starts at some time: for each chain $C$, we have clause $\bigvee x_{C,t}$, ranging over all starting times $t$ of chain $C$. We call those clauses \emph{chain clauses}. 
This forces that each chain has one of its variables to be true. 

The parameter $k$ is set to the number of chains. It follows that for each chain, exactly one of its variables is true. 

Then, we have clauses that check that we never use more than $m$ machines.
For each time $t'$, look at the set  \begin{align*}S_{t'} = \{x_{C,t} \mid \text{if $C$ starts at time $t$, then one of its jobs starts at time $t'$} \}.\end{align*} For each subset $X \subset S_{t'}$ of $m+1$ of these variables, we take a clause $\bigvee_{x_{C,t}\in X} \neg x_{C,t}$. We call those clauses the \emph{time clauses}. 

\begin{claim}
Suppose there is a solution for the \textsc{Weighted CNF-SAT} instance, then the \textsc{Chain Scheduling with Exact Delays} scheduling problem has a solution.
\end{claim}

\begin{proof}
Let $x_{C_1, t_1}, x_{C_2, t_2}, \ldots, x_{C_k, t_k}$ be a solution for the CNF-SAT instance. The chain clauses garantee that for every chain $C$ there is at least one variable $x_{C, t}$ true. Since $k$ equals the number of chains, we have that for every chain $C$ there is exactly one variable $x_{C, t}$ true. Let chain $C_i$ start at time $t_i$ for all $i = 1, 2, \ldots k$. 

We show that this is a feasible schedule by contradiction. Suppose that at time $t$ more than $m$ jobs are scheduled to start. Let $X$ be the set of chains that those jobs are in. Let $X' \subseteq X$ be a subset of size $m+1$. Then the time clause $\bigvee_{C_i\in X'} \neg x_{C_i,t_i}$ is not satisfied. This yields a contradiction. 
\end{proof}

\begin{claim}
If there exists a solution for the scheduling problem, then we have a solution for the \textsc{Weighted CNF-SAT} instance.
\end{claim}

\begin{proof}
Suppose we have a feasible schedule. For each chain $C$ set the variable $x_{C, t}$ to true, where $t$ is the starting time of $C$ in the schedule. Notice that there are exactly $k$ variables set to true. 

For every chain $C$ we set one variable $x_{C, t}$ to true, so the  chain clauses are satisfied. 

For every time $t$, at most $m$ jobs start at $t$. Consider the set $S_{t}$. At most $m$ variables in this set are set to true. So, for every subset $X \subseteq S_{t}$ of $m+1$ variables, there is at least one variable false. Hence, the clause is satisfied.  
\end{proof}

\subsection{Variable number of parallel machines}
\label{subsection:variablemachines}
In this section, we show that \textsc{Chain Scheduling with Exact Delays} problem with $m$ machines, where $m$ is part of the input, is 
$W[2]$-complete. 

\begin{theorem}
The \textsc{Chain Scheduling with Exact Delays} problems with a variable number of machines, parameterized by the number of chains $c$, is $W[2]$-complete. 
\label{thm:chains-m-machines-W2}
\end{theorem}

We prove this by reductions from and to \textsc{Dominating Set} problems. 

\begin{lemma}
The \textsc{Chain Scheduling with Exact Delays} problems with a variable number of machines, parameterized by the number of chains $c$, is $W[2]$-hard. 
\end{lemma}

\begin{proof}
Let  $G=(V,E)$, $k$ be an instance of \textsc{Dominating Set}. Assume that $V = \{v_1, v_2, \ldots, v_{n}\}$. 
We will make an instance of the scheduling problem with $k$ machines and $k+1$ chains. 

The first chain will have release date $2n-1$ and deadline $n(n+1)$. It will consist of $n$ jobs, with delay $n-1$ between each pair of consecutive jobs. Notice that this chain has to start at time $2n-1$. We will call this chain the \emph{check chain}. 

The other $k$ chains will be identical, we call them the \emph{vertex selection chains}. They have release date $0$ and deadline $(n+1)n + n - 1$. They will have total execution time $(n+1) n$, so they can start at times $0, 1, \ldots, n-1$. 
Start each chain with a job and then a delay of $n-1$. 

Consider the interval $[n, (n+1)n]$ from left to right. For each integer $in + j$ with $0 \leq j \leq n-1$, add the following to the vertex selection chains:
\begin{itemize}
    \item a delay of $1$ if $i = n-j$, 
    \item a delay of $1$ is $v_{i} \sim v_{n-j}$, 
    \item a job otherwise. 
\end{itemize}
Notice that for each interval $[in, (i+1)n]$, we made a gadget the represents the adjacencies of vertex $v_i$: there is no job at time $(i + 1)n - j$ of the execution of the chain if $i = j$ or $v_i \sim v_{j}$ for $j = 1, 2, \ldots, n$. 
See Figure \ref{fig:chains-parallel-machines}.

\begin{figure}
\includegraphics{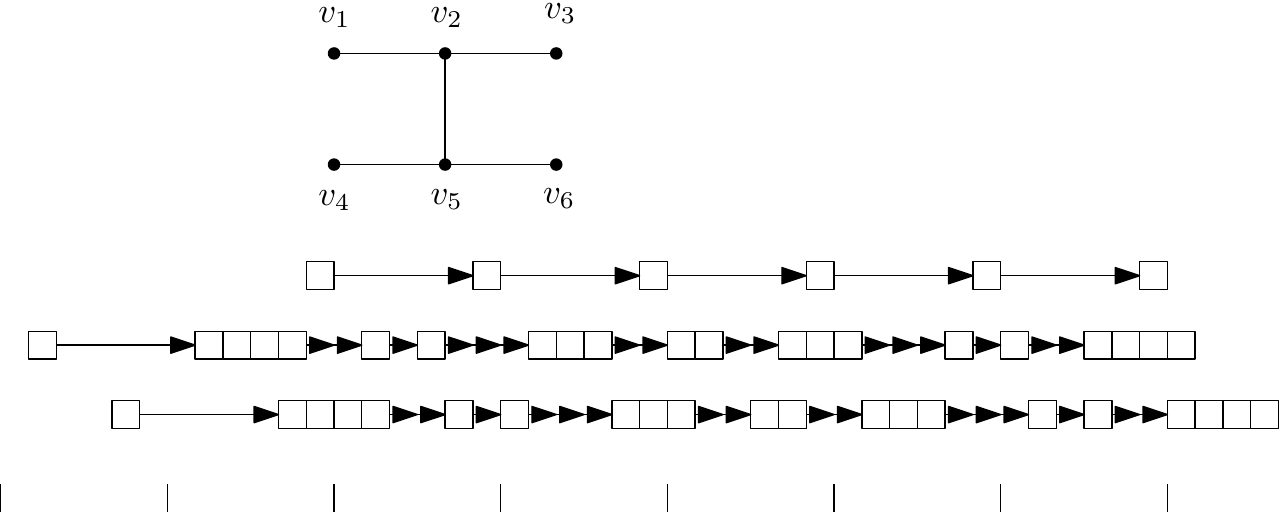}
\caption{A graph and the corresponding instance of the scheduling problem with $k = 2$.} \label{fig:chains-parallel-machines}
\end{figure}

\begin{claim}
If $G$ has a dominating set of size at most $k$, then there is a feasible schedule. 
\end{claim}
\begin{claimproof}
Let $v_{j_1}, v_{j_2}, \ldots, v_{j_k}$ be a dominating set. 
Let the vertex selection chains start at times $j_1 - 1$, $j_2 - 1$, $\ldots$, $j_k - 1$. We will show that this is a feasible schedule. 

Notice that at each time that is not of the form $in - 1$ at most $k$ jobs are scheduled to start, since the check chain has no job starting at this time. 

Now consider the time $in - 1$, for some $i = 1, 2, \ldots, n$. The check chain has a job scheduled to start at this time. We will show that there is a vertex selection chain that has no job starting at this time. 
Let $v_{j_a}$ be the vertex that dominates $v_i$. Then we know that $i = j_a$ or $v_i \sim v_{j_a}$. So the vertex selection chains have a delay of $1$ starting at time $(i+1) n - j_a$ of their execution time. Let $C$ be the chain that starts at time $j_a - 1$. It follows that $C$ has a delay of $1$ starting at time $(i+1) n - j_a + (j_a - 1) = (i+1) n - 1$. 
\end{claimproof}

\begin{claim}
If there is a feasible schedule, then $G$ has a dominating set of size at most $k$. 
\end{claim}
\begin{claimproof}
Let $j_1, j_2, \ldots, j_k$ be the starting times of the vertex selection chains in a feasible schedule. Consider the set $U = \{v_{j_1 + 1}, v_{j_2 + 1}, \ldots, v_{j_k + 1}\}$. We will show that this is a dominating set. 

Let $v_i$ be a vertex. Consider the time $(i+1)n - 1$. Notice that a job of the check chain is scheduled to start at this time, so one of the vertex selection chains $C$ has no job starting at this time. Let $j_a$ be the starting time of the chain $C$. It follows that $C$ has no job at time $(i+1)n - 1 - j_a$ of its execution time.  Hence $j_a + 1 = i$ or $v_{j_a + 1} \sim v_i$. We conclude that $v_i$ is dominated. 
\end{claimproof}
We conclude that the \textsc{Chain Scheduling with Exact Delays} problem with a variable number of machines is $W[2]$-hard, when parametrized by the number of chains. 
\end{proof}

\begin{lemma} \label{lem:chains-in-W2}
The \textsc{Chain Scheduling with Exact Delays} problems with a variable number of machines, parameterized by the number of chains $c$, is in $W[2]$. 
\end{lemma}
\begin{proof}
We use a reduction to {\sc Threshold Dominating Set}. In the {\sc Threshold Dominating Set} problem, we are given a graph $G$ and integers $k$ and $r$, and ask for a set of at most $k$ vertices, such that each vertex in $G$ is dominated at least $r$ times.

Downey and Fellows~\cite{DowneyF98} showed that {\sc Threshold Dominating Set}, parameterized by $k$ and $r$ is $W[2]$-complete. 

Suppose we have an instance of the \textsc{Chain Scheduling with Exact Delays} scheduling problem with $m$ machines and $c$ chains. Let $t_{\max}$ the largest deadline in this instance. 
We will distinguish three cases: $c-m \leq 0$, $c-m = 1$, and $c - m > 1$. If $c\leq m$, any schedule is feasible, so we reduce to a trivial dominating set instance. 

Suppose that $c - m =1$. 
We will construct an instance of \textsc{Threshold Dominating Set} where we look for a set of $c$ vertices that dominates each vertex at least once. 

We construct a graph $G$ with three types of vertices:
\begin{itemize}
    \item For each chain $C$ and each time step $t$, we have vertex $x_{C,t}$ if and only if it is possible to start chain $C$ at time $t$ (i.e., $t$ is not before the release date of $C$ and $t$ plus the execution time of $C$ is at most the deadline of $C$).
    \item We have an independent set with for each time $t$, with $0 \leq t< t_{\max}$, a vertex $y_t$.
    \item For each chain $C$ we have $c+1$ vertices $z_{z, \alpha}$, with $\alpha = 1, 2, \ldots, c+1$. 
\end{itemize}

We have the following edges:
\begin{itemize}
    \item For each chain $C$, the set of vertices of the form $x_{C,t}$ forms a clique.
    \item There is an edge between $x_{C,t}$ to $y_{t'}$, if and only if starting chain $C$ at time $t$ makes that no job of $C$ starts at time $t'$. 
    \item Each vertex $z_{C,\alpha}$ is adjacent to all vertices of the form $x_{C,t}$. 
\end{itemize}
See Figure \ref{fig:chains-in-W2-1}.

\begin{figure}
\centering
\includegraphics[height=.4\textheight]{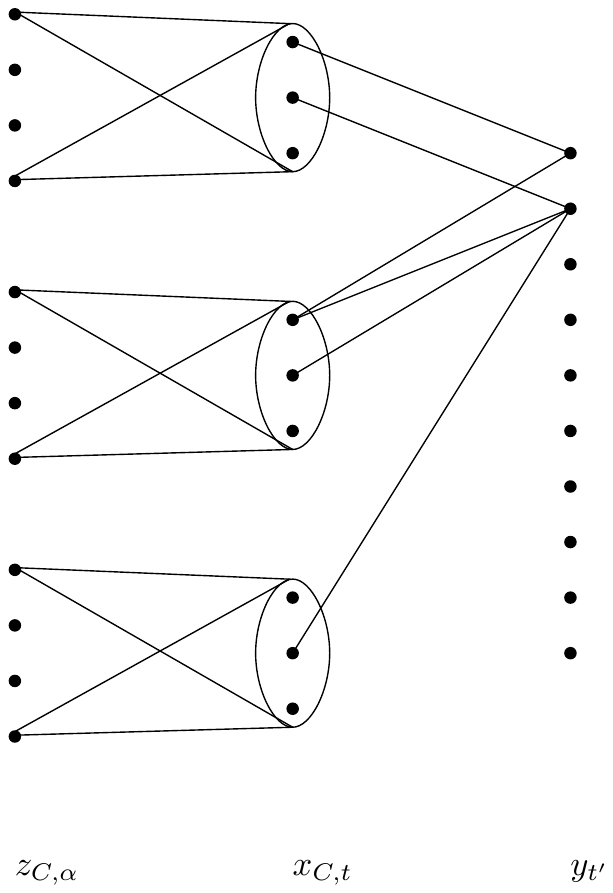}
\caption{The construction described in Lemma \ref{lem:chains-in-W2} for the case $c-m = 1$.} \label{fig:chains-in-W2-1}
\end{figure}

\begin{claim} \label{claim:chain-in-W2-1-a}
If $G$ contains a threshold dominating set of size at most $c$ that dominates every vertex at least once, then there is a feasible schedule. 
\end{claim}
\begin{claimproof}
Let $U$ be a threshold dominating set with $|W|\leq c$, such that each vertex is dominated at least once.

We will show that $U$ consists of vertices of the form $x_{C,t}$, one for each chain $C$. 

Consider a chain $C$. Notice that there are $c+1$ vertices $z_{C,\alpha}$, so $U$ does not contain all of them. Let $\alpha$ be such that $z_{C, \alpha} \notin U$. Since $z_{C, \alpha}$ is dominated by $U$, at least one vertex of the form $x_{C, t}$ is contained in $U$. Since we have $c$ chains and $U$ has size $c$, it follows that for every chain $C$, $U$ contains exactly one vertex of the form $x_{C,t}$ and no vertex of the form $z_{C, \alpha}$. Moreover, $U$ does not contain vertices of the form $y_t$.  

Make a schedule as follows: for every chain $C$, start this chain at time $t$, for $t$ such that $x_{C, t} \in U$. 
Consider a time $t'$. Since vertex $y_{t'}$ is dominated by a vertex $x_{C, t}$, we know that there is an edge from $x_{C, t}$ to $y_{t'}$. By the construction of $G$, it follows that chain $C$ has no job starting at time $t'$. It follows that at most $c-1 = m$ jobs start at time $t$. We conclude that this is a feasible schedule.  
\end{claimproof}

\begin{claim}
If there is a feasible schedule, then $G$ has a threshold dominating set of size at most $c$ that dominates every vertex at least once. 
\end{claim}
\begin{claimproof}
Consider the set $U$, that contains for every chain $C$ the vertex $x_{C,t}$, where $t$ is the starting time of $C$ in the schedule. 

It is clear that all vertices of the form $x_{C,t'}$ and $z_{C,\alpha}$ are dominated. 
Now consider a vertex $y_{t'}$. Since we have a feasible schedule, there is at least one chain that does not have a job starting at time $t'$. It follows that $y_{t'}$ is dominated by the corresponding vertex $x_{C,t}$. 
\end{claimproof}

Consider the last case $c-m > 1$. 

We will construct an instance of \textsc{Threshold Dominating Set} where we look for a set of $2c-m$ vertices that dominates each vertex at least $c-m$ times. 

We construct a graph $G$ with five types of vertices, this graph will be similar to the graph in the previous case. 
\begin{itemize}
    \item For each chain $C$ and each time step $t$, we have vertex $x_{C,t}$ if and only if it is possible to start chain $C$ at time $t$ (i.e., $t$ is not before the release date of $C$ and $t$ plus the execution time of $C$ is at most the deadline of $C$).
    \item We have an independent set with for each time $t$, with $0 \leq t< t_{\max}$, a vertex $y_t$.
    \item For each chain $C$ we have $c+1$ vertices $z_{z, \alpha}$, with $\alpha = 1, 2, \ldots, c+1$. 
    \item We have a clique of $c-m - 1$ vertices $w_i$, $1 \leq i \leq c-m - 1$.
    \item We have a vertex $v$. 
\end{itemize}

We have the following edges:
\begin{itemize}
    \item For each chain $C$, the set of vertices of the form $x_{C,t}$ forms a clique.
    \item There is an edge between $x_{C,t}$ to $y_{t'}$, if and only if starting chain $C$ at time $t$ makes that no job of $C$ starts at time $t'$. 
    \item Each vertex $z_{C,\alpha}$ is adjacent to all vertices of the form $x_{C,t}$. 
    \item As written above, the vertices $w_i$ form a clique.
    \item There is an edge $vw_i$ for every $i$. 
    \item Each vertex $w_i$ is adjacent to all vertices of the form $x_{C,t}$. 
    \item Each vertex $w_i$ is adjacent to all vertices of the form $z_{C,\alpha}$. 
\end{itemize}

See Figure \ref{fig:chains-in-W2-more}.
\begin{figure}
\centering
\includegraphics[height=.4\textheight]{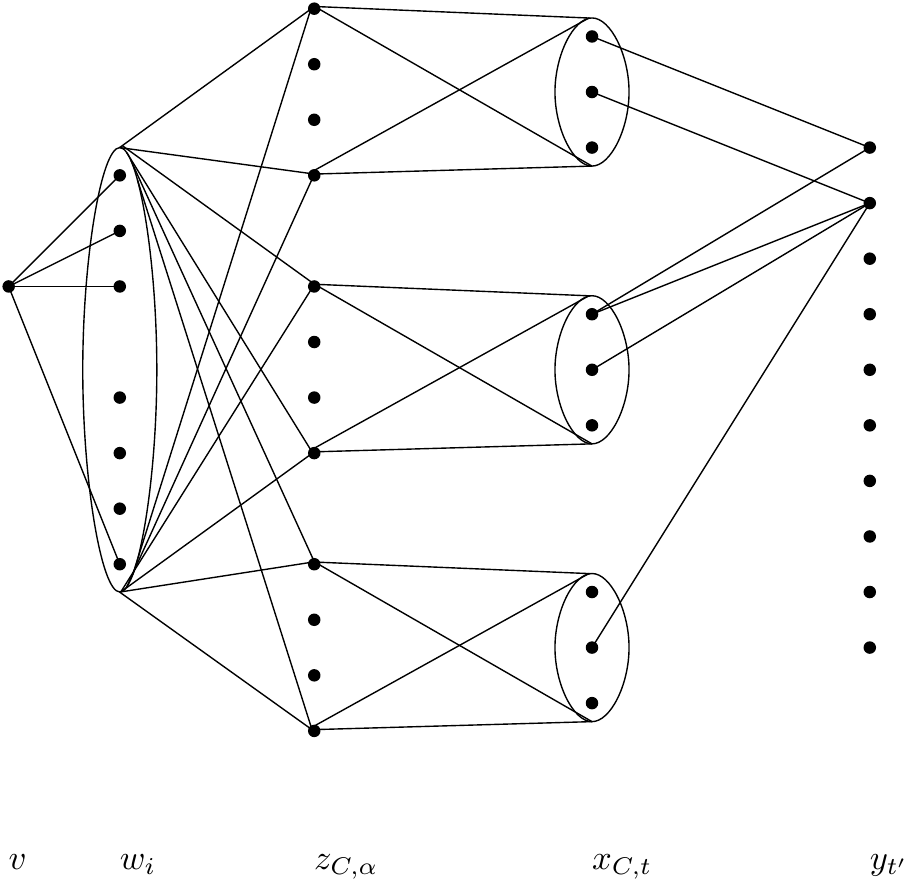}
\caption{The construction described in Lemma \ref{lem:chains-in-W2} for the case $c-m > 1$. The edges from $w_i$ to $x_{C,t}$ are omitted. } \label{fig:chains-in-W2-more}
\end{figure}

\begin{claim}
If $G$ has a threshold dominating set of size at most $2c-m$ that dominates every vertex at least $c-m$ times, then there is a feasible schedule. 
\end{claim}
\begin{claimproof}
Suppose that $G$ a threshold dominating set $U$ with $|U|\leq 2c-m$, such that each vertex is dominated at least $c-m$ times.

First, notice that, as $v$ has degree $c-m-1$, $U$ needs to contain $v$ and all its neighbours, i.e.,
all vertices of the form $w_i$. 
It follows that $U$ contains at most $c$ other vertices. 

As in the proof of Claim~\ref{claim:chain-in-W2-1-a}, it follows that $U$ contains exactly one vertex of the form $x_{C,t}$ for every chain $C$. 

Again as in the proof of Claim~\ref{claim:chain-in-W2-1-a}, we make a schedule by starting every chain at the time $t$ for which $x_{C,t} \in U$. 
Now consider a time $t'$. We know that $y_{t'}$ is dominated at least $c-m$ times. It follows that there are $c-m$ chains that have no job starting at time $t'$. We conclude that at most $m$ chains have a job starting at time $t'$, hence we use at most $m$ machines. 
\end{claimproof}

\begin{claim}
If there is a feasible schedule, then $G$ has a threshold dominating set of size at most $2c-m$ that dominates every vertex at least $c-m$ times. 
\end{claim}
\begin{claimproof}
Consider the set $U$ that contains $v$, all vertices $w_i$ and for every chain $C$ the vertex $x_{C,t}$, where $t$ is the starting time of $C$ is the schedule. 

It is clear that $v$ and every vertex $w_i$ is dominated at least $c-m$ times. Notice that every vertex $z_{C,\alpha}$ is dominated by every vertex $w_i$, and by one vertex $x_{C,t}$, so $z_{C, \alpha}$ is dominated at least $c-m$ times. The same holds for all vertices $x_{C,t}$. 

Now consider a vertex $y_{t'}$. We know that at most $m$ chains start a job at time $t'$. Hence, at least $c-m$ chains do not start a job at time $t'$. We conclude that $y_{t'}$ is dominated at least $c-m$ times. 
\end{claimproof}

We conclude that the \textsc{Chain Scheduling with Exact Delays} problem parametrized by the number of chains is in $W[2]$. 
\end{proof}

\subsection{Minimum delays}
\label{subsection:minimumtimeschains}

With a simple modification, the hardness proofs for exact delays in unary can be modified to hardness proofs for 
minimum delays (still in unary). The modification consists of taking a number of copies of the instance, as described below.

Suppose we have an instance for \textsc{Chain Scheduling with Minimum Delays} with exact delays with $c$ chains and $m$ machines. The $C$th chain has jobs
$j^C_1, j^C_2, \ldots, j^C_{s_C}$, with precedence constraints
$(j^C_i, j^C_{i+1})$, release date $r_C$ and deadline $d_C$. 

We define the {\em minimum duration} of a chain $C$ to be
$\ell_C = s_C + \sum_{i=1}^{s_C-1} l_{j^C_i, j^C_{i+1}}$. Thus, when the first job of
$C$ is executed at a time $t$, then the last job is executed at time $t+\ell_C -1$ or later, and thus the minimum duration $\ell_C$ denotes the minimum number of time steps from the first till the last execution of a job in the chain.

Set $T= \max_{1\leq C\leq c} d_C$ to be the maximum deadline of all chains.

We now build an instance for \textsc{Chain Scheduling with Minimum Delays}. 
The intuition behind the construction is the following: we build $cT+1$ identical copies of the original instance after each other. Each is executed in its own slot of $T$ consecutive time steps, in the same way as a solution of the original instance. If we have a solution of the new instance, then one of the copies has no preemption, \textit{i.e.}, must have delays equal to the minimum, and this gives a solution of the original instance.

For each chain $C$ from the original
instance, we make a new chain $C'$ in the new instance. This chain $C'$ has 
jobs $j^C_{i,a}$ with $1\leq i\leq s_C$ and $1\leq a \leq cT+1$.
We have precedence constraints $(j^C_{i,a}, j^C_{i+1,a})$ ($1\leq i <s_C$) whose minimum delay equals the
exact delay for the original jobs $(j^C_i, j^C_{i+1})$. The last job of the $a$th copy precedes the first job of the $(a+1)$st copy: we have a precedence constraint $(j^C_{s_C,a}, j^C_{1,a+1})$, for
$1\leq a < cT+1$. We set the minimum delay for this precedence in such a way that the minimum execution time between a job and its next copy equals $T$: the precedence constraint $(j^C_{s_C,a}, j^C_{1,a+1})$ goes with a minimum delay of $T-\ell_C$. The release date of the the new chain $C'$ equals $r_C$ and the deadline equals $cT^2+d_C$.

\begin{claim}
The instance of the \textsc{Chain Scheduling with Exact Delays} problem has a solution, if and only if the instance of the \textsc{Chain Scheduling with Minimum Delays} problem has a solution.
\end{claim}

\begin{proof}
Suppose we have a solution of the \textsc{Chain Scheduling with Exact Delays} problem with exact delays, where jobs $j^C_i$ are executed at time $t(j^C_i)$.
Now, schedule for the instance with minimum delays, jobs $j^C_{i,a}$ on time $t'(j^C_{i,a}) = (a-1)\cdot T + t(j^C_i)$. 

One can easily verify that this schedule fulfills the constraints. Note that all delays equal the minimum delays.

Now, suppose we have a solution of the \textsc{Chain Scheduling with Minimum Delays} problem with minimum delays. Note that the total execution time of 
a chain is at least $cT^2 + l_C$: the total execution time 
from a job $j^C_{1,a}$ till $j^C_{1,a+1}$ is constructed to be at least  
$T$, for $1\leq a < cT+1$. The total number of time steps from the 
release time till the 
deadline is 
$cT^2 + d_C - r_C$. 
Thus, the total slack in a chain, \textit{i.e.}, the sum of the differences between the scheduled delay time and the stated minimum delay, cannot be larger
than $cT^2 + d_C - r_C - cT^2 - l_C = d_C - r_C - l_C$. Notice that this is at most $T$. 

Say copy $a$ is {\em tainted} when there is at least one job $j^C_{i,a}$ for some $C$ and $i< s_C$ where
the delay from $j^C_{i,a}$ to $j^C_{i+1,a}$ is larger than the stated minimum delay. As each chain
has at most $T$ jobs of this type, we can have at most $cT$ tainted copies.

Thus, there is a copy $a$ that is not tainted. From the times that the jobs of copy $a$ are scheduled, we build a solution.

Notice that the first job $j^c_{1,1}$ is scheduled at the earliest at time $r_C$ and that $j^C_{1,a}$ is scheduled at least $(a-1)T$ time steps after $j^C_{1,1}$. So $j^C_{1,a}$ is scheduled at the earliest at time $(a-1)T + r_C$. Analogously, we find that $j^C_{s_C,a}$ is scheduled before $(a-1)T + d_C$.  Hence, every job $j^C_{i,a}$ is scheduled between $(a-1)T + r_C$ and $(a-1)T + d_C$. Suppose that job $j^C_{i,a}$ is scheduled to start at time $(a-1)T + t(j^C_{i,a})$, then schedule the corresponding job $j^C_i$ to start at time $t(j^C_{i,a})$. Notice that this is a feasible schedule, and that the exact delays are satisfied since the copy $a$ is not tainted. 
\end{proof}

By Theorems \ref{thm:chains-W1}, \ref{thm:chains-constant-machines-W1} and \ref{thm:chains-m-machines-W2}, we conclude the following results. 
\begin{theorem}
The \textsc{Chain Scheduling with Minimum Delays} problem with one machine, parameterized by the number of chains $c$, is $W[1]$-hard. 
\end{theorem}
\begin{theorem}
The \textsc{Chain Scheduling with Minimum Delays} problem with $m$ machines, parametrized by the number of chains, is $W[1]$-hard. 
\end{theorem}
\begin{theorem}
The \textsc{Chain Scheduling with Minimum Delays} problems with a variable number of machines, parameterized by the number of chains $c$, is $W[2]$-hard. 
\end{theorem} 

Since the tickness is upper bounded by the number of chains, the same results follow when parametrized by the tickness. 
\section{XP algorithms}
\label{sec:XP}

In this section, we give two positive results. First, we show membership in XP for all studied variants, when delays are given in unary, with a relatively straightforward dynamic programming algorithm. Then, in
the case of one machine, we show that {\sc Scheduling with Exact Delays} is in XP, when parameterized
by the number of chains, even when delays are given in binary.

\begin{lemma}
Given an instance of {\sc Chain Scheduling with Exact Delays} or {\sc Chain Scheduling with Minimum Delays},
one can build in polynomial time an equivalent instance, where all release dates and deadlines of chains
are nonnegative integers, bounded by $cn(D+1)$, where $c$ is the number of chains, $n$ the number of jobs,
and $D$ the maximum delay between two successive jobs in a chain. In addition, for each chain $C$, we 
have $d_C-r_C \leq n(D+1)$.
\label{lemma:unarydeadlines}
\end{lemma}

\begin{proof}
We first look at {\sc Chain Scheduling with Minimum Delays}. Consider the following operations. If we have a pair of successive jobs $j$ and $j'$ in a chain with minimum delay $d$, if
job $j$ is scheduled to start at time $t$, job $j'$ is scheduled to start at time $t'$, and there is a time $t''$
with $t+d +1 \leq t'' <t'$, and less than $m$ jobs are starting 
at time $t''$ (with $m$ the number of
machines), then we obtain a valid schedule when we reschedule job $j'$ at $t''$ and do not change the time
for any other job. If the first job in a chain $C$ is scheduled at time $t$, and there is a time step
$t'$ with $r_C \leq t' < t$, then we obtain a valid schedule when we reschedule this first job at time
$t'$ and no not change the time for any other job.
Call a schedule {\em left-adjusted} when these steps are not possible. If there
is a valid schedule, then there is also a valid left-adjusted schedule: just take any valid schedule,
and perform the steps above while possible.

In a left-adjusted schedule, for each chain $C$, there is a job (of this or another chain) starting at time
$r_C$, and we cannot have an interval with $D +1$ successive time steps between $r_C$ and the scheduled time
of the last job in the chain without any job (of any chain) scheduled. Thus, the last job of a chain
is scheduled to start at time 
$r_C + n + (n-1)D - 1\leq r_C + (n-1)(D+1)$ or earlier. This shows that we obtain an equivalent
instance if we set for all chains the deadline to 
$\min\left\{ d_C, r_C+(n-1)(D+1) +1 \right\}$. Since $r_C+(n-1)(D+1) +1 \leq r_C+ n(D+1)$, we obtain an equivalent instance if we set all deadlines to $\min\{d_C, r_C+n(D+1)\}$.

In the case of {\sc Chain Scheduling with Exact Delays}, valid schedules are trivially left-adjusted,
so setting deadlines to 
$\min\left\{ d_C, r_C+n(D+1) \right\}$ again gives an equivalent instance.

\smallskip

Now, consider the interval graph where each chain is a vertex, with edges between vertices if the chains overlap, here, 
the intervals associated with chain vertices $C$ are the intervals $[r_C,d_C)$.
We consider the different connected components of this interval graph. Take the component that contains
the chain with minimum release date, and call this $r_{\min}$. For each chain in the component,
subtract $r_{\min}$ from its release date and deadline. We obtain an equivalent instance, with at
least one release date equal to 0. We say that this first component is {\em handled}.

We now handle the other components, one by one as follows. Set $\delta$ to be the maximum deadline of
all chains in handled components. Take, among all components that are not yet handled, the one with
minimum release date of a chain in it. Say this minimum release date is $r_{\min '}$. Subtract from
all chains in the component $r_{\min '} - \delta$ 
from the release date and deadline. Thus, the 
earliest release date in the component is now $\delta$.
It is easy to see that this gives an equivalent schedule.

The algorithm thus consists of two steps: first, we set deadlines to $\min\left\{ d_C, r_C+n(D+1) \right\}$, 
and then we handle the components. Now, order the chains by release date: $C_1, C_2, \ldots, C_c$. Notice that $r_{C_{i+1}} \leq d_{C_i}$. Since $d_{C_i} \leq r_C + n(D+1)$, it follows that the maximum deadline is at most $cn(D+1)$.  
\end{proof}

\begin{theorem}
{\sc Chain Scheduling with Exact Delays} and {\sc Chain Scheduling with Minimum Delays} belong to XP,
when delays are given in unary, and parameterized by the number of chains or thickness, for any number of machines.
\end{theorem}

\begin{proof}
As the thickness is never larger than the number of chains, it suffices to give the result for the thickness.

The first step is to build the equivalent instance with deadlines and release dates in unary,
and for all chains $d_C-r_C \leq n(D+1)$, as in Lemma~\ref{lemma:unarydeadlines}.

We give now the dynamic programming algorithm; there are only small differences between the
algorithms for {\sc Chain Scheduling with Exact Delays} and {\sc Chain Scheduling with Minimum Delays}.

The algorithm uses the notion of a {\em state}. A state
consists of a time step $t$, and for each chain $C$
with $t\in [r_C,d_C)$, we have a bool started$(C)$. If $\started(C) = \textit{true}$, then the state also contains  
a job $j_C$ and a time $t_C \in [r_C,t]$.
(The intuition is as follows: $j_C$ is the last job of chain $C$ that is scheduled at or before time $t$;
$j_C$ is scheduled at time $t_C$. If no job of $C$ is scheduled at or before time $t$, then 
{\em started}$(C)$ is {\em false}.)

We say a state is {\em possible}, when there is a schedule with the following properties:
\begin{itemize}
    \item The schedule assigns a time to all jobs in a chain whose deadline is before $t$, and for all
    chains $C$ with $t\in [r_C,d_C)$ and started$(C)$ is true, a time is assigned to all jobs in the chain
    from the first job up to job $j_C$ --- i.e., we do not assign times for jobs preceded by $j_C$.
    \item Jobs $j_C$ are scheduled at time $t_C$, for all chains $C$ with $t\in [r_C,d_C)$ and started$(C)$ is true.
    \item The schedule respects the conditions on (minimum or exact) delays, and number of machines.
\end{itemize}

The dynamic programming algorithm now consists of computing for each time step $t$, in order of increasing $t$, the set of possible states for time $t$. Let $D$ be the maximum over all chains $C$ of $d_C-r_C$.
Note that the number of states at time $t$ is bounded by $(1+ n^2(D+1))^\tau$: each of the at
most $\tau$ chains $C$ with $t \in [r_C, d_C)$, we either have started$(C)$ false, or select one of the
(at most $n$) jobs in the chain, and one of the (at most $n(D+1)$) timesteps in $[r_C,d_C)$.
So, for fixed $\tau$, the number of states is polynomial in the input size.

It is not hard to see, that given a set of all possible states at time $t$, we can compute in polynomial
time the set of all possible states at time $t+1$. (The details are somewhat tedious: for all possible
states at time $t$, consider all decisions of the form where for
each machine we decide if we schedule a job on this machines, and if so, what job. For each of these, check whether delay and deadline conditions are still fulfilled; if so, add this to the possible states for 
time $t+1$. The check is different for exact or minimum delays; apart from that, these variants are handled in the same way.)

Continue computing this set for the time step that equals the maximum deadline minus $1$. If it has a possible state, then this state reflects a schedule where all jobs are scheduled, and we decide positively; otherwise, no valid schedule exists.
\end{proof}

\begin{theorem}
{\sc Chain Scheduling with Exact Delays} with $m=1$, parameterized by the number of chains, with delays in binary belongs to XP.
\end{theorem}

\begin{proof}
Suppose we have chains $C_1, \ldots, C_c$. Suppose chain $C_i$ has jobs $j_{i,1}$, $j_{i,2}$, $\ldots$, $j_{i, \ell_i}$, with $j_{i,a}$ directly preceding $j_{i,a+1}$; we write the exact delay of this constraint as $l_{i,a}$. Write $s(i,a) = \sum_{b=1}^{a-1} (l_{i,b}+1)$. Note that $j_{i,a}$ has to be scheduled exactly
$s(i,a)$ time after $j_{i,1}$ starts.

For each chain $C_i$, we take a variable $x_i$ that denotes the time that the first job of $C_i$ is scheduled. 

\begin{claim}
Variables $x_1, x_2, \ldots, x_c$ give a valid schedule, if and only if the following constraints are fulfilled.
\begin{enumerate}
    \item For each $i$, $x_i \geq r_{C_i}$.
    \item For each $i$, $x_i + s(i, \ell_i) < d_{C_i}$. 
    \item For each $i$ and $i'$ with $i\neq i'$, and each $j$, $j'$ with $1\leq j \leq \ell_i$,
    $1 \leq j' \leq \ell_{i'}$, we have $x_i + s(i,j) \neq x_{i'} + s(i',j')$.
\end{enumerate}
\label{claim:algorithmdifference}
\end{claim}

\begin{proof}
Suppose we have a valid schedule where chain $C_i$ starts at time $x_i$. As the first job in a chain
does not start before the release date, we have $x_i \geq r_{C_i}$. As the last job of the chain
starts at time $x_i + s(i, \ell_i)$, we have $x_i + s(i, \ell_i) < d_{C_i}$. 
If the third condition would not hold for a 4-tuple $i$, $i'$, $j$, $j'$, then both the $j$th job of
chain $C_i$ and the $j'$th job of $C_{i'}$ would be scheduled at time $x_i + s(i,j)$; this gives a conflict as we have only one machine.

The other direction is (also) simple: the first condition ensures that chains do not start before the release date; the second that they finish before the deadline, and the third that no two jobs are scheduled at the same time.
\end{proof}

The first step of the algorithm is to compute for each pair $i$, $i'$ with $i\neq i'$ a set
$U(i,i') = \left\{ s(i',j') - s(i,j)
~|~ 1\leq j \leq \ell_{i},~ 1 \leq j' \leq \ell_{i'} \right\}$.
Note that each of these sets has size $O(n^2)$, or more precisely, is at most the product of the sizes of the two chains.
Now, sort each set $U(i,i')$.

Suppose $U(i,i') = \{a_1, a_2, \ldots, a_r\}$ with $a_1 < a_2 < \cdots < a_r$. Condition 3 of Claim~\ref{claim:algorithmdifference} for the pair $i$, $i'$ can be expressed as
\begin{multline*}
 \left(x_i-x_{i'} < a_1\right) \vee \left(a_1 < x_i-x_{i'} < a_2 \right) \vee \left(a_2 < x_i-x_{i'} < a_3 \right) \vee \cdots \\ \cdots \vee \left(a_{r-1} < x_i-x_{i'} < a_r\right) \vee \left(a_r < x_i-x_{i'}\right) 
\end{multline*}

Our algorithm now branches on these $O(n^2)$ possibilities. For each of the $O(c^2)$ pairs of chains,
we have $O(n^2)$ branches, which gives a total of $O(n^{O(c^2)})$ subproblems.

Each of these subproblems asks to solve a set of inequalities. These inequalities are of the form
$x_i - x_{i'} < a$  or
$x_i \geq a$ (Condition 1 of Claim~\ref{claim:algorithmdifference}) or $x_i \leq a$ (Condition 2 of Claim~\ref{claim:algorithmdifference}), for some integers $a$. As we work with integers and look for integer solutions,
we reformulate constraints of the form $x_i - x_{i'} < a$ as $x_i-x_{i'} \leq a-1$. 
We now have a system of linear inequalities which can be solved in polynomial time with text book (shortest paths) methods,
see e.g., \cite[Section 24.4]{CormenLRS}. If at least one of the subproblems has a solution, then this
solution gives starting times for the chains that gives a valid schedule; otherwise, there is no valid schedule.

We have $O(n^{O(c^2)})$ branches, each taking polynomial time, and this gives a running time of $O(n^{O(c^2)})$.
\end{proof}

\section{Conclusions}
\label{section:conclusions}
In this paper, we have shown a number of results on the parameterized complexity of
{\sc Chain Scheduling with Exact Delays} and {\sc Chain Scheduling with Minimum Delays}.
In a few cases, we obtained $W[1]$-completeness or $W[2]$-completeness; in the other cases,
we only showed hardness results, often together with XP-membership. We expect that the problems, parameterized
by the thickness do not belong to $W[P]$ --- for the same `compositionality' reason as why one can
believe that {\sc Graph Bandwidth} does not belong to $W[P]$: see the discussion in
\cite[Section 4]{FellowsR20}. The machinery to prove such results currently is not available, but
we conjecture that also the variants with minimum delays inhibit some form of compositionality and do
not belong to $W[P]$.

We end this paper with mentioning some open problems. In this paper, we proved for the case that
delays are given in binary, for only one of the cases membership in XP. What is the complexity of
the other cases when delays are given in binary? Also, an interesting question is to study the
variant where we have maximum delays, with all of its subcases.


\begin{thebibliography}{10}

\bibitem{Aigner2001}
Martin Aigner and G{\"u}nter~M. Ziegler.
\newblock Bertrand's postulate.
\newblock In {\em Proofs from {THE BOOK}}, pages 7--12. Springer, 2001.
\newblock \href {http://dx.doi.org/10.1007/978-3-662-04315-8_2}
  {\path{doi:10.1007/978-3-662-04315-8_2}}.

\bibitem{Bessy2019}
S.~Bessy and R.~Giroudeau.
\newblock Parameterized complexity of a coupled-task scheduling problem.
\newblock {\em Journal of Scheduling}, 22(3):305--313, 2019.
\newblock \href {http://dx.doi.org/10.1007/s10951-018-0581-1}
  {\path{doi:10.1007/s10951-018-0581-1}}.

\bibitem{BodlaenderF95}
Hans~L. Bodlaender and Michael~R. Fellows.
\newblock {$W[2]$}-hardness of precedence constrained {$K$}-processor
  scheduling.
\newblock {\em Oper. Res. Lett.}, 18(2):93--97, 1995.
\newblock \href {http://dx.doi.org/10.1016/0167-6377(95)00031-9}
  {\path{doi:10.1016/0167-6377(95)00031-9}}.

\bibitem{Brucker}
Peter Brucker, Johann Hurink, and Wieslaw Kubiak.
\newblock Scheduling identical jobs with chain precedence constraints on two
  uniform machines.
\newblock {\em Mathematical Methods of Operations Research}, 49:211--–219,
  1999.
\newblock \href {http://dx.doi.org/10.1007/PL00020913}
  {\path{doi:10.1007/PL00020913}}.

\bibitem{CormenLRS}
Thomas~H. Cormen, Charles~E. Leiserson, Ronald~L. Rivest, and Clifford Stein.
\newblock {\em Introduction to Algorithms, 3rd Edition}.
\newblock {MIT} Press, 2009.
\newblock URL: \url{http://mitpress.mit.edu/books/introduction-algorithms}.

\bibitem{DowneyF95}
Rodney~G. Downey and Michael~R. Fellows.
\newblock Fixed-parameter tractability and completeness {I}: {B}asic results.
\newblock {\em SIAM Journal on Computing}, 24:873--921, 1995.
\newblock \href {http://dx.doi.org/10.1137/S0097539792228228}
  {\path{doi:10.1137/S0097539792228228}}.

\bibitem{DowneyF95II}
Rodney~G. Downey and Michael~R. Fellows.
\newblock {F}ixed-parameter tractability and completeness {II}: {O}n
  completeness for {$W[1]$}.
\newblock {\em Theoretical Computer Science}, 141(1-2):109--131, 1995.
\newblock \href {http://dx.doi.org/10.1016/0304-3975(94)00097-3}
  {\path{doi:10.1016/0304-3975(94)00097-3}}.

\bibitem{DowneyF98}
Rodney~G. Downey and Michael~R. Fellows.
\newblock Threshold dominating sets and an improved characterization of
  {$W[2]$}.
\newblock {\em Theoretical Computer Science}, 209(1-2):123--140, 1998.
\newblock \href {http://dx.doi.org/10.1016/S0304-3975(97)00101-1}
  {\path{doi:10.1016/S0304-3975(97)00101-1}}.

\bibitem{ErdosT41}
P.~Erd\"{o}s and P.~Tur\'{a}n.
\newblock On a problem of {S}idon in additive number theory, and on some
  related problems.
\newblock {\em Journal of the London Mathematical Society}, s1-16(4):212--215,
  1941.
\newblock \href {http://dx.doi.org/10.1112/jlms/s1-16.4.212}
  {\path{doi:10.1112/jlms/s1-16.4.212}}.

\bibitem{Fellows2003}
Michael~R. Fellows and Catherine McCartin.
\newblock On the parametric complexity of schedules to minimize tardy tasks.
\newblock {\em Theoretical Computer Science}, 298(2):317--324, 2003.
\newblock \href {http://dx.doi.org/10.1016/S0304-3975(02)00811-3}
  {\path{doi:10.1016/S0304-3975(02)00811-3}}.

\bibitem{FellowsR20}
Michael~R. Fellows and Frances~A. Rosamond.
\newblock Collaborating with {H}ans: Some remaining wonderments.
\newblock In Fedor~V. Fomin, Stefan Kratsch, and Erik~Jan van Leeuwen, editors,
  {\em Treewidth, Kernels, and Algorithms --- Essays Dedicated to Hans L.
  Bodlaender on the Occasion of His 60th Birthday}, volume 12160 of {\em
  Lecture Notes in Computer Science}, pages 7--17. Springer, 2020.
\newblock \href {http://dx.doi.org/10.1007/978-3-030-42071-0\_2}
  {\path{doi:10.1007/978-3-030-42071-0\_2}}.

\bibitem{Mnich2015}
Matthias Mnich and Andreas Wiese.
\newblock Scheduling and fixed-parameter tractability.
\newblock {\em Mathematical Programming}, 154(1):533--562, 2015.
\newblock \href {http://dx.doi.org/10.1007/s10107-014-0830-9}
  {\path{doi:10.1007/s10107-014-0830-9}}.

\bibitem{wikum1994one}
Erick~D. Wikum, Donna~C. Llewellyn, and George~L. Nemhauser.
\newblock One-machine generalized precedence constrained scheduling problems.
\newblock {\em Operations Research Letters}, 16(2):87--99, 1994.
\newblock \href {http://dx.doi.org/10.1016/0167-6377(94)90064-7}
  {\path{doi:10.1016/0167-6377(94)90064-7}}.

\bibitem{WOEGINGER}
Gerhard~J. Woeginger.
\newblock A comment on scheduling on uniform machines under chain-type
  precedence constraints.
\newblock {\em Operations Research Letters}, 26(3):107--109, 2000.
\newblock \href {http://dx.doi.org/10.1016/S0167-6377(99)00076-0}
  {\path{doi:10.1016/S0167-6377(99)00076-0}}.

\end{thebibliography}
\end{document}